\documentclass[twocolumn,envcountsect]{svjour-alt}

\usepackage[numbers]{natbib} 
\usepackage{paralist}        
\usepackage{graphicx}        
\usepackage{algorithm}       
\usepackage[noend]{algpseudocode}
\usepackage[small,compact]{titlesec}
\usepackage{amsmath}
\usepackage{amssymb}

\makeatletter
\renewcommand\bibsection%
{
  \section*{\refname
    \@mkboth{\MakeUppercase{\refname}}{\MakeUppercase{\refname}}}
}
\makeatother

\journalname{Distributed Computing}

\renewenvironment{proof}[1][Proof.]{\begin{trivlist}
\item[\hskip \labelsep {\emph{#1}}]}{\qed\end{trivlist}}

\newcommand{\half}{\frac{1}{2}}
\newcommand{\floor}[1]{\left\lfloor #1 \right\rfloor}
\newcommand{\set}[1]{\left\{#1\right\}}
\newcommand{\bigO}{\mathcal{O}}
\newcommand{\paren}[1]{\left( #1 \right)}
\newcommand{\PR}[2][]{\Pr_{#1}\left[#2\right]}
\newcommand{\st}{\mbox{ s.t. }}

\def\beep{{beep} }

\begin{document}
\title{Beeping a Maximal Independent Set}

\author{Yehuda Afek \and Noga Alon \and Ziv
    Bar-Joseph \and\\ Alejandro Cornejo \and Bernhard
    Haeupler \and Fabian Kuhn\\
}
\authorrunning{Afek, Alon, Bar-Joseph, Cornejo, Haeupler, Kuhn} 
\institute{%
Y. Afek \at The Blavatnik School of Computer Science,\\
Tel Aviv University, 69978, Israel
\and
N. Alon \at Sackler School of Mathematics,\\
Tel Aviv University, 69978, Israel
\and
Z. Bar-Joseph \at School of Computer Science,\\
Carnegie Mellon Univ., Pittsburgh, PA 15213, USA
\and
A. Cornejo \and B. Haeupler \at Computer Science and Artificial Intelligence 
Laboratory,\\
Massachusetts Institute of Technology, MA 02139, USA
\and
F. Kuhn \at Department of Computer Science,\\
University of Freiburg, 79110 Freiburg, Germany
}
\date{November 18, 2011}

  \maketitle 

  \begin{abstract}


We consider the problem of computing a maximal independent set (MIS) in an
extremely harsh broadcast model that relies only on carrier sensing.  The model
consists of an anonymous broadcast network in which nodes have no knowledge
about the topology of the network or even an upper bound on its size.
Furthermore, it is assumed that an adversary chooses at which time slot 
each node wakes up.  At each time
slot a node can either beep, that is, emit a signal, or be silent. At a 
particular
time slot, beeping nodes receive no feedback, while silent nodes can only
differentiate between none of its neighbors beeping, or at least one of 
its neighbors beeping.

We start by proving a lower bound that shows that in this model, it is not
possible to locally converge to an MIS in sub-polynomial time.
We then study four different relaxations of the model which allow us
to circumvent the lower bound and find an MIS in polylogarithmic
time. First, we show that if a polynomial upper bound on the network
size is known, it is possible to find an MIS in $\bigO(\log^3 n)$
time. Second, if we assume sleeping nodes are awoken by neighboring 
beeps, then
we can also find an MIS in $\bigO(\log^3 n)$ time. Third, if in
addition to this wakeup assumption we allow sender-side collision detection,
that is, beeping nodes can distinguish whether at least one neighboring node
is beeping concurrently or not, we can find an MIS in $\bigO(\log^2 n)$ time.
Finally, if instead we endow nodes
with synchronous clocks, it is also possible to find an MIS in
$\bigO(\log^2 n)$ time.



  \keywords{Maximal Independent Set\and Distributed\and Beeps\and Radio 
  Networks\and Asynchronous Wakeup}
  \end{abstract}

  \section{Introduction}

An MIS is a maximal set of nodes in network such that no two
nodes in the set are neighbors. Since the set is maximal every node in 
the
network is either in the MIS or has a neighbor in the MIS.  The
problem of distributively finding an MIS has been extensively
studied in various models
\cite{alon86,collier1996pat,Peleg,KMNW05,KMW06,KMW04,luby86,podc05,dist11,Z10} 
and has many
applications in networking, and in particular in radio sensor
networks.  Some of the practical applications include the
construction of a backbone for wireless networks, as a foundation for
routing and for clustering, and for generating spanning trees
to reduce communication costs \cite{Peleg,Z10}.

This paper studies the problem of finding an MIS in the discrete
beeping wireless network model introduced in~\cite{beepcolor}. The 
network is
modeled as an undirected graph and time progresses in discrete and synchronous
time slots.  In each time slot a node can either transmit a
``jamming'' signal (called a beep) or detect whether at least one of its
neighbors beeps.  We believe that such a model is minimalistic enough
to be implementable in many real world scenarios.  For example, it can easily
be implemented using carrier sensing alone, where nodes only differentiate
between silence and the presence of a signal on the wireless channel.  Further,
it has been shown that such a minimal communication model is strong enough to
efficiently solve non-trivial tasks
\cite{science11,beepcolor,infocom09,schneider10disc}.  The model is interesting
from a practical point of view since carrier sensing typically uses less energy
to communicate and reaches larger distances when compared with sending regular
messages.


While this model is clearly useful for computer networks, it is also 
useful to model biological processes.
In
biological systems, cells communicate by secreting certain proteins
that are sensed (``heard") by neighboring cells
\cite{collier1996pat}. This is similar to a node in a radio network
transmitting a carrier signal which is sensed (``heard") by its
neighbors. Such physical message passing allows for an upper bound
on message delay. Thus, for a computational model based on these
biological systems, we can assume a set of synchronous and anonymous
processes communicating using beeps~\cite{beepcolor} in an
arbitrary topology. We have recently shown that a variant of MIS is
solved by a biological process, sensory organ precursor (SOP)
selection in flies, and that the fly's solution provides a novel
algorithm for solving MIS \cite{science11}. Here we extend
algorithms for this model in several ways as discussed below.

The paper has two parts. First we prove a lower bound that shows
that in a beeping model with adversarial wake-up it is not possible
to locally converge to an MIS in sub-polynomial time. Next we present 
several relaxations of this model under which
polylogarithmic MIS constructions are possible.

The lower bound shows that if nodes are not endowed with any
information about the underlying communication graph, and their
wake-up time is under the control of the adversary, any (randomized)
distributed algorithm to find an MIS requires at least
$\Omega(\sqrt{n/\log n})$ rounds.  We remark that this lower bound
holds much more generally. We prove the lower bound for the
significantly more powerful radio network model with collision detection 
and arbitrary message sizes. The lower bound is therefore not an 
artifact of the amount of information which can be communicated in the 
beeping model.

Following the lower bound, in the second part of the paper four weaker 
models are considered and a polylogarithmic time algorithm for
an MIS construction is presented for each of these models. First, we 
present an
algorithm that uses a polynomial upper bound on the size of the
network, to compute an MIS in $\bigO(\log^3 n)$ rounds with high
probability. Our next two algorithms assume that nodes are awakened
by incoming beeps (wake-on-beep). First, we present an $\bigO(\log^2 n)$ 
rounds algorithm in the wake-on-beep model with sender collision 
detection. 
Next, we present a $\bigO(\log^3 n)$ time algorithm that works without
sender collision detection in the same wake-on-beep model. Finally, we 
show that even if nodes are only by an adversary (and \emph{not} by 
incoming beeps) it is possible to use synchronous clocks to compute an 
MIS in $\bigO(\log^2 n)$ time without any information about the network.  
The results are summarized in Table~\ref{table:results}.  We highlight 
that all the upper bounds presented in this paper compute a stable MIS 
eventually and almost surely. That is, once an MIS is computed it is 
stable and the probability that no MIS is computed until time $t$ is 
exponentially small in $t$. Thus only the running times of our 
algorithms are randomized.

\begin{table}
\caption{Model restrictions and algorithmic running times}
\label{table:results}
  \centering
\begin{tabular}{c|p{1.4in}|c}
  Section & Assumptions & Running Time \\ \hline
  4 & None (lower bound) & $\Omega(\sqrt{n/\log n})$ \\
  5 & Upper bound on $n$ & $\bigO(\log^3 n)$\\
  6 & Wake-on-Beep + Sender Collision Detection & $\bigO(\log^2 n)$\\
  7 & Wake-on-Beep & $\bigO(\log^3 n)$\\
  8 & Synchronous Clocks & $\bigO(\log^2 n)$\\
\end{tabular}
\end{table}

  \section{Related Work}

The problem of finding an MIS has been recognized and studied as a
fundamental distributed computing problem for a long time (e.g.,
\cite{alon86,awerbuch89,luby86,panconesi95}). Perhaps the single
most influential MIS algorithm is the elegant randomized algorithm
of \cite{alon86,luby86}, generally known as Luby's algorithm, which
has a running time of $\bigO(\log n)$. This algorithm works in a
standard message passing model, where nodes can concurrently and
reliably send and receive messages over all point-to-point links to
their neighbors. \citet{dist11} showed how to improve the bit
complexity of Luby's algorithm to use only $\bigO(\log n)$ bits per
channel ($\bigO(1)$ bits per round).  For the case where the size of the
largest independent set in the $2$-neighborhood of each node is
restricted to be a constant (known as bounded independence or
growth-bounded graphs), \citet{podc08} presented an algorithm that
computes an MIS in $\bigO(\log^* n)$ rounds. This class of graphs
includes unit disk graphs and other geometric graphs that have been
studied in the context of wireless networks.

While several methods were suggested for computing an MIS in a
distributed setting, most previous algorithms are designed for a
classical message passing model without message interference and
collisions and they are based on the assumption that nodes know
something about the local or global topology of the network. The
first effort to design a distributed MIS algorithm for a wireless
communication model in which the number of neighbors is not known is
by \citet{mass04}. They provide an algorithm for the radio network
model with a $\bigO(\log^9 n/\log\log n)$ running time. This was later
improved~\cite{podc05} to $\bigO(\log^2 n)$. Both algorithms assume
that the underlying graph is a unit disk graph (the algorithms also
work for somewhat more general class of geometric graphs). In
addition, while the algorithms solve the MIS problem in multi-hop
networks with adversarial wake up, they assume that an upper bound on
the number of nodes in the network is known. In addition to the upper
bound assumption their model allows for (and their algorithm uses)
messages whose size is a function of the number of nodes in the
network.

The use of carrier sensing and collision detection in wireless
networks has been studied in
\cite{chlebus00,ilcinkas10,schneider10disc}. As shown in
\cite{schneider10disc}, collision detection can be powerful and can
be used to improve the complexity of algorithms for various basic
problems. \citet{mobihoc08} show how to approximate a minimum
dominating set in a physical interference (SINR) model where in
addition to sending messages, nodes can perform carrier sensing. In
\cite{flurywattenhofer}, it is demonstrated how to use carrier
sensing as an elegant and efficient way for coordination in
practice.

The present paper is not the first one that uses carrier sensing alone
for distributed wireless network algorithms. A similar model to the
beeping model considered here was first studied in
\cite{degesys07,infocom09}. As used here, the model has been
introduced in \cite{beepcolor}, where it is shown how to efficiently
obtain a variant of graph coloring that can be used to schedule
non-overlapping message transmissions.  In \cite{science11} a variant
of the beeping model, there called the fly model, was considered. The
fly model makes three additional assumptions: that all the processes
wake up at the same round, that a bound on the
network size is known to the processes, and that senders can detect
collisions. That is, processes can listen on the medium while
broadcasting (as in some radio and local area networks).  Apart from
\cite{science11}, the most closely related work to this paper are
results from \cite{schneider10disc}. In \cite{schneider10disc}, it is
shown that in growth-bounded graphs (a.k.a.\ bounded independence 
graphs) an MIS can be computed in
$\bigO(\log n)$ time using only carrier sensing. Specifically, they 
assume nodes have receiver-side collision detection, they know the 
polynomial growth function of the graph, they known an upper bound on 
the size of the network and they have unique identifiers.
The present paper studies the MIS problem in general graphs under the 
beeping model.

  \section{Model}
\label{sec:model}

Following  \cite{beepcolor}, we consider a synchronous communication network
modeled by an arbitrary graph $G = (V,E)$ where the vertices $V$ 
represent processes and the edges
represent pairs of processes that can hear each other.
We denote the set of neighbors of
node $u$ in $G$ by $N_G(u)=\set{v \mid \set{u,v} \in E}$. For a node
$u \in V$ we use $d_G(u)=|N_G(u)|$ to denote its degree (number of
neighbors) and we use $d_{\max}=\max_{u \in V} d_G(u)$ to denote the
maximum degree of $G$.

Initially all processes are asleep, and a process starts participating in the 
round after it is woken up by an adversary.  We denote by $G_t \subseteq G$
the subgraph induced by the processes which are participating in round $t$.

Instead of communicating by exchanging messages, 
we consider a more primitive communication model that relies
entirely on carrier sensing. Specifically, in every round a
participating process can choose to either beep or listen.
If a process $v$ listens in round $t$
it can only distinguish between silence (i.e., no process $u \in
N_{G_t}(v)$ beeps in round $t$) or the presence of one or more beeps
(i.e., there exists at least one process $u \in N_{G_t}(v)$ that beeps in round
$t$).
Observe that a beep conveys less information than a conventional
1-bit message, for which it is possible to distinguish
between no message, a message with a one, and a message with a zero.

Given an undirected graph $H$, a set of vertices $I \subseteq V(H)$
is an independent set of $H$ if every edge $e \in E(H)$ has at most one
endpoint in $I$. An independent set $I \subseteq V(H)$ is a maximal
independent set of $H$, if for all $v \in V(H)\setminus I$ the set
$I\cup\set{v}$ is not independent.

An event is said to occur with
high probability, if it occurs with probability at least $1-n^{-c}$
for any constant $c \ge 1$, where $n=|V|$ is the number of nodes in
the underlying communication graph. For a positive integer $k \in
\mathbb{N}$ we use $[k]$ as short hand notation for
$\set{1,\ldots,k}$. In a slight abuse of this notation we use $[0]$
to denote the empty set $\varnothing$ and for $a,b \in \mathbb{N}$
and $a < b$ we use $[a,b]$ to denote the set $\set{a,\ldots,b}$.

During the execution of an algorithm each node may go through several different 
states. Of particular interest are the inactive-state and the MIS-state, which 
are present in all the algorithms described in this paper.
A node is defined as \emph{stable} if it is in the MIS-state and all its
neighbors are in the inactive-state, or if it has a stable neighbor in the 
MIS-state.
Observe that by definition, if all nodes are stable then all nodes are either 
in the MIS-state or in the inactive-state. In all our algorithms, 
once a node becomes stable it remains stable thereafter, and moreover 
eventually all nodes become stable with probability one.
We will prove that the algorithms we propose guarantee that with high 
probability nodes becomes stable quickly and the nodes which are in the 
MIS-state describe a maximal independent set.


Specifically we say that a (randomized) distributed algorithm the MIS 
problem in $T$ rounds if, when no additional nodes are woken up for $T$ 
rounds, the nodes which are in the MIS-state
describe a stable MIS (with high probability). Furthermore, we require that eventually the nodes which are
in the MIS-state describe a stable MIS with probability one.
Additionally, we say an algorithm locally converges to an MIS in $T$ 
rounds, if any node (with high probability) irrevocably decides $T$ 
rounds (regardless of wakeups) whether to be in the MIS-state or not.



\section{Lower Bound for Uniform Algorithms}
\label{sec:lowerbound}

In this section we show that without any additional power or a priori
information about the network, e.g., an upper bound on its size or
maximum degree, any randomized distributed algorithm that locally 
converges to an MIS needs at least polynomial time.

We stress that this lower bound is not an artifact of the beeping model,
but a limitation that stems from having message
transmission with collisions and the fact that nodes are required to decide (but not
necessarily terminate) without waiting until all nodes have woken up 
(i.e., locally converge).
Although we prove the lower bound for the problem of finding an
MIS, the bound can be generalized to other problems (e.g., minimal
dominating set, coloring, etc.).

Specifically, we prove the lower bound for the much stronger 
communication
model of local message broadcast with collision detection.
In this model a process can choose in every round either
to listen or to broadcast a message (no restrictions are made on the
size of the message). When listening a process receives silence if
no message is broadcast by its neighbors, it receives a collision
if a message is broadcast by two or more neighbors, and it receives
a message if it is broadcast by exactly one of its neighbors.
The beep communication model can be easily simulated by this model
(instead of beeping send a $1$ bit message, and when listening translate a
collision or the reception of a message to hearing a beep) and hence the
lower bound applies to the beeping model.

At its core, our lower bound argument relies on the observation that a node can 
learn essentially no information about the graph $G$ if after waking up,
it always hears collisions or silence. It thus has to decide whether it
remains silent or beeps within a \emph{constant} number of rounds. More
formally:

\begin{proposition} \label{prop:beep}
  Let $\mathcal{A}$ be an algorithm run by all nodes, and consider a fixed 
  pattern $H \in
  \{\mathrm{silent},\mathrm{collision}\}^*$.
  If after waking up a node $u$ hears $H(r)$ whenever it listens in
  round $r$, then there are two constants $\ell \geq 1$ and $p \in
  (0,1]$ that depend on only $\mathcal{A}$ and $H$ such that either
  \begin{inparaenum}[\bf a)]
  \item $u$ remains listening indefinitely, or 
  \item $u$ listens for $\ell-1$ rounds and broadcasts in round $\ell$
  with probability $p$.
  \end{inparaenum}
\end{proposition}

\begin{proof}
  We fix a node $u$ and let $p(r)$ be the probability with which node $u$
  beeps in round $r$. Observe that $p(r)$ can only depend on $r$, what
  node $u$ heard up to round $r$, that is, $H[1\ldots r]$ and its random coin 
  flips.
  Therefore, given any algorithm, either $p(r)=0$ for all $r$ (and node
  $u$ remains silent forever), or $p(r) > 0$ for some $r$, in which case
  we let $p=p(r)$ and $\ell=r$.
\end{proof}

We now prove the main result of this section:

\begin{theorem}
  \label{thm:lowerbound}
  If nodes have no a priori information about the graph $G$ then any
  distributed algorithm in the local message broadcast
  model with collision detection that locally converges to an MIS 
  requires with constant probability at least $\Omega(\sqrt{n/\log
  n})$ rounds.
\end{theorem}

\begin{proof}
We fix any algorithm $\mathcal{A}$ and use Proposition~\ref{prop:beep} to split 
the
analysis into three cases. In all cases we show that there is a family
of graphs on which, with probability $1-o(1)$, algorithm $\mathcal{A}$ 
does not locally converge to an MIS if it is run for $o(\sqrt{n/\log 
n})$ rounds.
  
We first ask what happens with nodes running algorithm $\mathcal{A}$ that hear
only silence after waking up.
Proposition~\ref{prop:beep} implies that either nodes remain silent
forever, or there are constants $\ell$ and $p$ such that nodes broadcast
after $\ell$ rounds with probability $p$.
In the first case, suppose nodes are in a clique, and observe that no
node will ever broadcast anything. In this case nodes cannot
learn anything about the underlying graph and in particular cannot break 
symmetry between them. Thus, either no node joins the MIS, or all
nodes join the MIS independently with constant probability, in which case
their success probability is exponentially small in $n$.

We can thus apply Proposition~\ref{prop:beep} and assume for the rest of the
argument that nodes running $\mathcal{A}$ that hear only silence after waking
up broadcast after $\ell$ rounds with probability $p$.
Now we consider what happens with nodes running $\mathcal{A}$ that hear
only collisions after waking up. Again, by
Proposition~\ref{prop:beep} we know that either they remain silent
forever, or there are constants $m$ and $p'$ such that nodes 
broadcast after $m$ rounds with probability $p'$.
In the rest of the proof we describe a different execution for each of
these cases.

\begin{figure*}[htpb]
  \centering
  \includegraphics{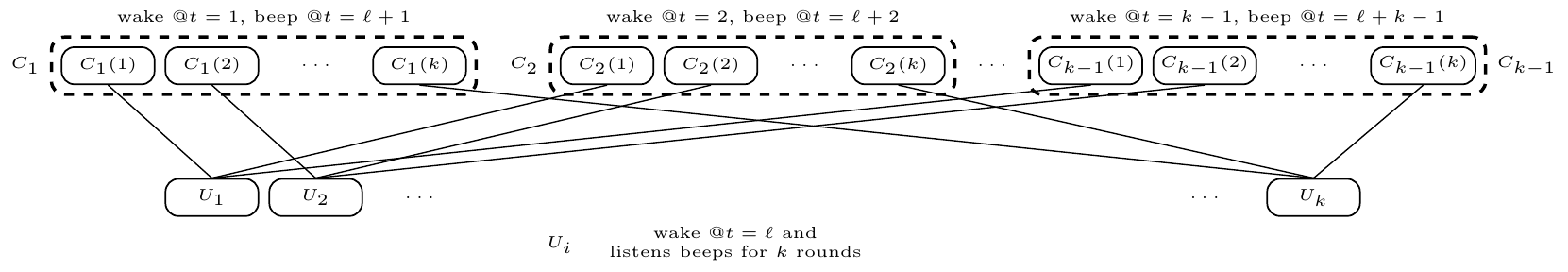}
  \caption{Execution for Case 1}
\end{figure*}

\paragraph{CASE 1: (a node that hears only collisions remains silent
forever)}

We consider a network topology consisting of several interconnected
node-disjoint cliques. For some $k \gg \ell$ to be fixed later, we
take a set of $k-1$ cliques $C_1,\ldots,C_{k-1}$ and a set of $k$
cliques $U_1,\ldots,U_{k}$, where each clique $C_i$ has $\Theta(k\log
n/p)$ vertices, and each clique $U_j$ has $\Theta(\log n)$ vertices.
We consider a partition of each clique $C_i$ into $k$ sub-cliques
$C_i(1),\ldots,C_i(k)$ where each sub-clique has $\Theta(\log n/p)$ 
vertices.  For
simplicity we say two cliques are connected if they form a complete 
bipartite graph.

For every $j\in[k]$ clique $U_j$ is connected to
sub-clique $C_i(j)$ for each $i\in[k-1]$. We consider the execution
where in round $i \in [k-1]$ clique $C_i$ wakes up, and in round
$\ell$ the cliques $U_1,\ldots,U_k$ wake up simultaneously.  Hence,
when clique $U_j$ wakes up, it is connected to sub-clique $C_i(j)$ for
each $i < \ell$.  Similarly for $i\geq \ell$, when clique $C_i$ wakes
up, for all $j \in [k]$, sub-clique $C_i(j)$ is connected to clique
$U_j$.

Because the first nodes wake up in round $1$, no node participates in
round $1$. During the rounds $2,\dots,\ell$, only the nodes in
$C$-cliques are participating and they all remain silent and hear
silence.  In round $\ell+1$ every node in $C_1$ broadcasts with
probability $p$.  Thus with high probability for all $j \in [k]$ at
least two nodes in sub-clique $C_1(j)$ broadcast in round $\ell$. This
guarantees that all the nodes in the $U$-cliques hear a collision
during the first round they are awake, and hence they also listen for
the second round. In turn, this implies that the nodes in $C_2$ hear
silence during the first $\ell-1$ rounds they participate, and again
for $j \in [k]$, with high probability, there are at least two nodes
in $C_2(j)$ that broadcast in round $\ell+2$.

We can extend this argument inductively to show that, with
high probability. For each $i \in [k-1]$ and for every $j \in [k]$ at
least two nodes in sub-clique $C_i(j)$ broadcast in round $\ell+i$.
Therefore, with high probability, all nodes in cliques
$U_1,\ldots,U_k$ hear collisions during the first $k-1$ rounds after
waking up.

Observe that at most one node in each $C_i$ clique can join the MIS, 
that is,
only one of the sub-cliques of $C_i$ has a node in the MIS. Since
there are more $U$-cliques than there are $C$-cliques the pigeon
hole principle implies that there exists at least one clique $U_j$
that is connected to only non-MIS nodes. However, since the nodes in
$U_j$ are connected in a clique, exactly one node of $U_j$ must decide
to join the MIS. Note that all nodes in $U_j$ have the same state
during the first $k-1$ rounds. Therefore, if nodes decide after
participating for at most $k-1$ rounds, with constant probability, 
either no node in $U_j$ joins the MIS, or more than two nodes join the MIS.

Finally since the number of nodes $n$ is $\Theta(k^2\log n+k\log n)$,
we can let $k \in \Theta(\sqrt{n/\log n})$ and the claim follows.

\begin{figure}[htpb]
  \centering
  \includegraphics{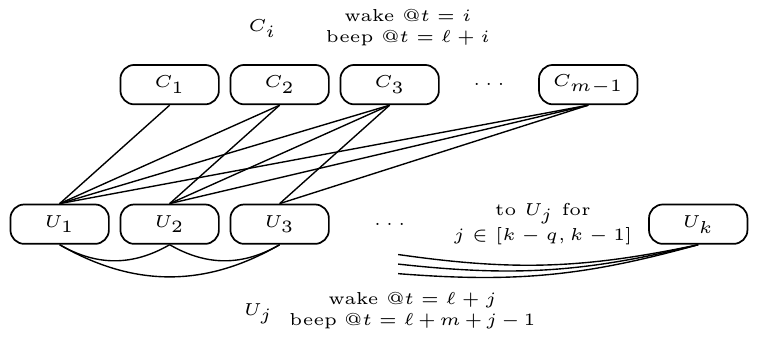}
  \caption{Execution for Case 2}
\end{figure}


\paragraph{CASE 2: (after hearing only collisions, a node beeps with
    probability \boldmath$p'$ after $m$ rounds)}

For some $k \gg m$ to be fixed later let $q=\floor{\frac{k}{4}}$ and
consider a set of $k$ cliques $U_1,\ldots, U_k$ and a set of $m-1$
cliques $C_1,\ldots,C_{m-1}$, where each clique $U_i$ is of size 
$\Theta(\log
n/p')$, and each clique $C_i$ is of size $\Theta(\log n/p)$.  As before, 
we say two cliques are connected if they form a
complete bipartite graph.

If $j > 1$ then $U_j$ is connected to every $U_i$ for
$i \in \set{\max(1,j-q),\ldots,j-1}$ and if $j < m$ then $U_j$ is 
connected to every clique $C_h$ for $h \in \set{j,\ldots,m}$.
We consider the execution where in round $i \in [m-1]$ clique $C_i$ 
wakes up, and in round $\ell+j$ for $j \in [k]$ clique $U_j$ wakes up.

During the rounds $2,\dots,\ell-1$, the nodes in $C_1$ are
participating without hearing anything else, and hence every node in
$C_1$ broadcasts in round $\ell+1$ with probability $p$. Therefore, with 
high probability, at least two nodes in $C_1$ broadcast in round
$\ell+1$.  This guarantees the nodes in $U_1$ hear a collision after
waking up at round $\ell+1$, and therefore they listen in round 
$\ell+2$. In turn this
implies the nodes in $C_2$ will also hear silence during the first
$\ell-1$ rounds they participate, and hence, with high probability, at
least two nodes in $C_2$ broadcast in round $\ell+2$.

As before, we can extend this execution inductively to show that for
$i \in [m-1]$ the nodes in $C_i$ hear silence for the first $\ell-1$
rounds they participate, and, with high probability, at least two
nodes in $C_i$ broadcast in round $\ell+i$.  Moreover, for $j \in [k]$
the nodes in $U_j$ hear collisions for the first $m-1$ rounds they
participate, and hence with high probability there are at least two
nodes in $U_j$ who broadcast in round $\ell+m+j-1$. This implies that
with high probability for $j \in [k-q]$ the nodes in $U_j$ hear
collisions for the first $q$ rounds they participate.

We show that if nodes choose whether or not to join the MIS $q$ rounds
after participating, then they fail with high probability. In
particular consider the nodes in clique $U_j$ for $j \in
\set{q,\ldots,k-2q}$.  These nodes will hear collisions for the first
$q$ rounds they participate, and they are connected only to other nodes
which also hear collisions for the first $q$ rounds they participate.
Therefore, if nodes decide after participating for at most $q$ rounds,
with constant probability either a node and all its neighbors will not 
be in the MIS, or two or more neighboring nodes join the MIS.

Finally since we have $n \in \Theta(m\log n + k\log n)$ nodes, we can
let $k \in \Theta(n/\log n)$ and hence $q \in \Theta(n/\log n)$ and
the theorem follows.
\end{proof}

\subsection{Termination Lower Bound}\label{sec:termination}

In this section we provide a basic observation about
symmetry breaking in beep networks. We conclude that no MIS algorithm
for this model can safely terminate at any time. This justifies why all 
our algorithms guarantee safety by running indefinitely.

We note that for the same reasons as before, the results in this
subsection hold even for local message broadcast with collision
detection. Moreover they hold under the assumptions that nodes wake up
at the same time and know the size of the network.  This includes the
knowledge of an upper bound on the size of the network assumed in
Section~\ref{sec:alg1}, the wake-on-beep model in 
Section~\ref{sec:beepmis} and the
assumption of synchronized clocks in Section~\ref{sec:synch}.  It 
applies thus to all our algorithms, except
the one in Section~\ref{sec:Nmis}.

\begin{lemma}\label{lem:indistinguish}
It is impossible for a node to distinguish at time $t$
with probability more then $1 - 2^{-t+1}$ between an execution in which it
is in isolation and an execution in which it has exactly one neighbor.  
\end{lemma}
\begin{proof}
Initially nodes start in identical states and the probability of
distinguishing between being isolated (1-node graph) or having one 
neighbor (2-node graph) is at most $1/2$. In each round the symmetry 
between two neighboring nodes is broken only if one node beeps while the 
other node listens. Since the nodes are assumed to be in identical 
states they both have the same probability $p$ to beep. Therefore the 
probability that symmetry is broken for the first time in any particular 
round is $\max_p 2(1-p)p \leq 1/2$. Hence the probability that after $t$ 
rounds the nodes remain in identical states is at least $2^{-t}$.  
Finally this implies that the probability that after $t$ rounds a node 
cannot distinguish between an execution in which it is isolated and an 
execution where it has exactly one neighbor is at most $1 - 2^{-t+1}$.
\end{proof}

\begin{lemma}\label{lem:notermination}
    An algorithm that solves the MIS problem cannot terminate with a
    correct solution in every execution.
\end{lemma}
\begin{proof}
  In a 1-node graph a node must join the MIS while in a connected 2-node 
  graph exactly one node must join the MIS and the other must not.
  Therefore, a node cannot terminate if it cannot distinguish between 
  these two cases. Finally, by Lemma \ref{lem:indistinguish} at any time 
  $t$ there is a non-zero probability for any algorithm to not being 
  able to distinguish between these two cases.
\end{proof}

\begin{lemma}\label{lem:lognlowerbound}
There are graphs where the expected time for any algorithm to converge 
to a stable MIS is at least $\frac{\log n}{e}$.
\end{lemma}
\begin{proof}
Consider a graph of $n/2$ disjoint pairs of neighboring nodes.  A MIS 
algorithm can only terminate if it has broken the
symmetry in each component where breaking the symmetry is independent 
between components.
By Lemma~\ref{lem:indistinguish} for any $t$ the probability for this to 
happen is at most $(1 - 2^{-t})^{n/2}$. Using Markov's inequality the 
expected time until all pairs break symmetry is at least $k (1 - 
2^{-k})^{n/2}$ for any $k$. Setting $k = \log n$ shows that the expected 
time to compute an MIS is at least $\log n (1-\frac{1}{n})^{n/2}$ and 
for $n\ge 2$ this implies the expected time to compute an MIS is at 
least $\log n/e$.
\end{proof}

All our algorithms will always (eventually) converge to a stable MIS, 
and with high probability they converge to an MIS in polylogarithmic 
time.
Lemma \ref{lem:notermination} and Lemma \ref{lem:lognlowerbound}
imply that both properties are best possible.  A Las Vegas algorithm is 
one which always produces the correct output but whose running time is 
probabilistic. Conversely, a Monte Carlo algorithm is one whose running 
time is deterministic but only produces the correct output with high 
probability. As we pointed out, all the algorithms presented in this 
paper are Las Vegas. However by assuming an upper bound on $n$ it is 
possible to turn any of these algorithms into a Monte Carlo algorithm.  
Specifically,  it suffices to add an early-stopping criteria (using the 
upper bound on $n$) once the output is correct with high probability.
Another alternative to convert these Las Vegas algorithms to Monte Carlo 
is to endow nodes with unique identifiers.  Specifically, using these 
identifiers it is possible to augment the algorithms to detect the case 
where two neighboring nodes are in the MIS state with certainty in 
asymptotically the same round complexity as the bit length of the 
identifiers. Yet another alternative to circumvent Lemma 
\ref{lem:indistinguish} is to assume sender-side collision detection, an 
assumption which we consider detail in Section \ref{sec:Nmis}. This 
allows two nodes to detect in a single round whether they are beeping 
alone, or if there is another neighbor beeping. Our algorithm in 
Section~\ref{sec:Nmis} leverages this assumption to terminate after 
$\bigO(\log^2 n)$ rounds.


  \section{Using an Upper Bound on \boldmath$n$} \label{sec:alg1}

\def\plength{c \log N}

In this section we give an example demonstrating that knowing a priori 
information about the network can drastically
change the complexity of the problem. More precisely we show that
by giving all nodes a (crude) upper bound $N > n$ on the total number of 
nodes participating in the system, it is possible to circumvent the 
polynomial lower bound for Section~\ref{sec:lowerbound} and design an 
algorithm that locally converges to an MIS in polylog time.
It is not required that all nodes are given the same upper bound.
We will describe an algorithm that guarantees that 
$\bigO(\log^2 N \log n)$ rounds after a node wakes up, it knows whether 
it belongs to the MIS or if it  has a neighbor in the MIS.  This implies 
that if the known upper bound is polynomial in $n$ its possible to 
design an algorithm that locally converges to an MIS in $\bigO(\log^3 
n)$ rounds.

\paragraph{Algorithm.}

If at any point during the execution a node hears a beep while listening it 
restarts the algorithm.
When a node wakes up (or it restarts), it stays in an inactive state
where it listens for $c\log^2 N$ consecutive rounds.
After this inactivity period, nodes enter a competing state where rounds are 
grouped
into $\log N$ phases of $\plength$ consecutive rounds. Observe that due 
to the adversarial wake up and the restarts, the phases of different 
nodes may not be synchronized.
In each round of phase $i$, a node beeps with probability $2^i/8 N$, and
otherwise it listens. Therefore by phase $\log N$ a node beeps
with constant probability in every round.
After successfully going through the $\log N$ competing phases
(recall that when a beep is heard during any phase, the algorithm
restarts) a node assumes it has joined the MIS and goes into an infinite loop 
where it beeps half of the time to claim its MIS status while listening 
the rest of the time to detect if a neighboring node is also in the MIS. 

\begin{algorithm}
  \caption{Upper bound on the size of the network.}
  \label{alg:CHKmis1}
  \begin{algorithmic}[1]
    \State Restart here whenever receiving a beep.
    \State {\bf for} $c \log^2 N$ rounds {\bf do} listen \Comment{Inactive}
    \For{$i \in \set{1,\ldots,\log N}$} \Comment{Competing}
      \For{$\plength$ rounds}
      \State with probability $2^i/(8 N)$ beep, otherwise listen
      \EndFor
    \EndFor
    \State {\bf forever at each round}
    \State \quad with probability $\half$ beep then listen
    \State \quad else listen then beep
    \Comment{MIS}
  \end{algorithmic}
\end{algorithm}

\begin{theorem}\label{thm:simpleMIS}
    If $N$ is an upper bound on $n$ known to the nodes, 
    Algorithm~\ref{alg:CHKmis1} locally converges to an MIS in $O(\log^2 
    N \log n)$ rounds.
\end{theorem}

We remark that Algorithm~\ref{alg:CHKmis1} is very robust. It is not hard to show that
it is self-stabilizing, that is, nodes can be initialized in any state and with any 
setting of internal variables without affecting the guarantees. It also works as-is 
under adversarial crashes, that is, if we give the adversary the power to
crash any set of nodes in every round. However, in the presence of 
crashes, no algorithm can locally converge to an MIS, since an inactive 
node with a
single neighboring MIS node cannot always immediately join the MIS when its MIS 
neighbor crashes.
Nevertheless, Algorithm~\ref{alg:CHKmis1} computes an MIS in $O(\log^2 N \log n)$ rounds.
We also refer to the discussion in Section~\ref{sec:termination} which 
shows that without additional assumptions this Las Vegas algorithm can 
be turned into a Monte Carlo algorithm that with high probability gives 
a correct answer and always terminates in $O(\log^3 N)$ steps.

\paragraph{Safety.}

We first prove the safety property of Algorithm~\ref{alg:CHKmis1} in the following lemma:

\begin{lemma}\label{lem:safety}
  Two neighboring nodes do not join the MIS with high probability. Moreover, in 
  the low probability event that two neighboring nodes join the MIS, then almost 
  surely eventually one of them becomes inactive.
\end{lemma}
\begin{proof}
Observe that a node must go through an interval of at least $\plength$ rounds in which it is both
listening and beeping with constant probability in every round. If a node is competing while
another node is in the MIS or if two nodes are competing for the MIS in their last phase
both nodes need to choose the same action (beep or listen) for 
$\plength$ rounds in order for both nodes to be in the MIS state 
simultaneously. Therefore, for a sufficiently large constant $c$ this 
event will not happen with high probability. On the other hand, even if 
two neighboring nodes join the MIS, the probability that they both 
remain in the MIS after $k$ rounds is exponentially small in $k$, so it 
follows that eventually almost surely one of the nodes will leave the 
MIS.  \end{proof}

We note that by construction nodes which are in the MIS beep at 
least every three rounds. Hence, if a node is in the MIS and all its neighbors 
are inactive, it follows that the MIS node and its neighbors will remain 
stable indefinitely (or until a neighbor crashes).

\paragraph{Termination.}

Given Lemma~\ref{lem:safety} it remains to show the following lemma to finish the proof of Theorem~\ref{thm:simpleMIS}:

\begin{lemma}
  \label{lem:bterminate}
  With high probability after $\bigO(\log^2 N\log n)$ rounds a node is 
  either in the MIS or has a neighbor in the MIS.
\end{lemma}

We prove this in three steps. First we show that for any node, the sum 
of the beep probabilities of its neighbors cannot increase ``quickly'' 
after $\plength$ rounds.
We then use this to to show that when a node $u$ is competing, then with 
constant probability the sum of the beep probabilities of the neighbors 
of $u$ are less than a constant.
Finally, we show that a node $u$ hears a beep or produces a beep every 
$O(\log^2 N)$ rounds. Every time this happens there is a constant 
probability that either a neighbor of $u$ joins the MIS or that $u$
joins the MIS. Therefore, with high probability the algorithm produces an 
MIS after $O(\log^2 N \log n)$ rounds.

First we introduce some additional definitions. We use $b_u(t)$ to 
denote the \emph{beep probability} of
node $u$ in round $t$.  The \emph{beep potential} of a set of nodes $S
\subseteq V$ in round $t$ is defined as the sum of the beep
probabilities of nodes in $S$ in round $t$, and denoted by
$E_S(t)=\sum_{u \in S} b_u(t)$.  Of particular interest is the beep
potential of the neighborhood of a node, we will use $E_v(t)$ as a
shorthand notation for $E_{N(v)}(t)$.

The next lemma shows that if the beep potential of a particular set of
nodes is larger than a (sufficiently large) constant in a particular 
round, then it was larger than a constant in the preceeding $\plength$ 
rounds. Informally, this is true because the beep probability of every 
node
increases slowly.

\begin{lemma}
  \label{lem:slowchange}
  Fix a set $S \subseteq V$. If $E_S(t) \ge \lambda$ in round $t$,
  then $E_S(t') \ge \half\lambda-\frac{1}{8}$ for all $t' \in
  [t-\plength,t]$.
\end{lemma}

\begin{proof}
    First we define a partition of the nodes in $S$. Let $P\subseteq S$ 
    be the nodes in $S$ that are in phase $1$ at round $t$, let $Q$ be 
    the set of nodes which are in phase $i > 1$ at round $t$, and let 
    $R$ be the remaining nodes (i.e., the ones which are not competing).  
    By definition the nodes in $R$ do not contribute to the beep 
    potential of the nodes in $S$, we have:
    
    \[
    E_S(t)=\underbrace{\sum_{u \in P} b_u(t)}_{E_P(t)}+ 
    \underbrace{\sum_{u \in Q} b_u(t)}_{E_Q(t)}
    \]

    Fix $t'$ to be any round in the range $[t-\plength,t]$. Since
    nodes in $P$ are in phase $1$ in round $t$, in round $t'$ they are
    either in the inactive state or in phase $1$. Thus for $u \in P$
    we have $b_u(t') \le b_u(t)=1/(4 N)$, and since there are at most
    $|P| \le |S| \le N$ nodes, we have $E_P(t') \le E_P(t)= (N/4) N =
    1/4$.

  Similarly nodes in $Q$ are in phase $i > 1$ in round $t$ and in phase 
  $i$ or $i-1 \ge 1$ in round $t'$.
  Thus for $u \in Q$ we have $b_u(t') \ge \half b_u(t)$ and hence 
  $E_Q(t') \ge \half E_Q(t) = \half(E_S(t)-E_P(t)) \ge
  \half\lambda-\frac{1}{8}$.
  Finally since $E_S(t') \ge E_Q(t')$, we get $E_S(t') \ge
  \half\lambda-\frac{1}{8}$.
\end{proof}

Using the previous lemma, we show that with high probability nodes that
are competing have neighborhoods with a ``low'' beep potential.
Informally this is because if a node had a neighborhood with a ``high''
beep potential, the previous result implies it would have had a high beep
potential during the previous $c\log N$ rounds, and therefore with high 
probability it would have been kicked out of the competition in a 
previous round.

\begin{lemma}
  \label{lem:lowEu}
  With high probability, if node $v$ is competing in round $t$ then $E_v(t) <
  \half$.
\end{lemma}

\begin{proof}
  Fix a node $v$ and a time $t$, we will show that if $E_v(t)\ge \half$
  then with high probability node $v$ is not competing at time $t$.
  
  Let $L_v(\tau)$ be the event that node $v$ listens in round $\tau$ and
  there is a neighbor $u \in N(v)$ that beeps in round $\tau$. First we
  estimate the probability of the event $L_v(\tau)$.

  \begin{align*}
    \PR{L_v(\tau)} &= (1-b_v(\tau))\cdot\paren{1-\prod_{u \in
    N(v)}(1-b_u(\tau))} \\ & \ge
    (1-b_v(\tau))\cdot\paren{1-\exp\paren{-\sum_{u\in N(v)} b_u(\tau)}} \\
    &= (1-b_v(\tau))\cdot\left(1-e^{-E_v(\tau)}\right).
  \end{align*}

  From Lemma~\ref{lem:slowchange} we have that if $E_v(t)\ge \half$
  then $E_v(\tau) \ge \frac{1}{8}$ for $\tau \in [t-\plength, t]$,
  together with the fact that $b_v(\tau) \le \half$ this implies that
  
  \[
  \PR{L_v(\tau)} \ge \half\paren{1-e^{-1/8}} > 0.058
  \] 
  for $\tau \in [t-\plength,t]$.

  Let $C_v(t)$ be the event that node $v$ is competing in round $t$.
  Observe that if $L_v(\tau)$ occurs for $\tau \in [t-\plength,t]$
  then node $v$ stops competing for at least $\plength$ rounds and hence
  $C_v(t)$ cannot occur. Therefore, the probability that node $v$ does
  not beep in round $t$ is at least:

  \begin{align*}
    \PR{\neg C_v(t)} &\ge \PR{\exists \tau \in
    [t-\plength,t] \st L_v(\tau) \text{ occurs}}\\ &\ge 1-\prod_{\tau=t-\plength}^t 
    (1-\PR{L_v(\tau)})\\
    &\ge 1-\exp\paren{-\sum_{\tau=t-\plength}^t \PR{L_v(\tau)}}.
  \end{align*}
  
  Finally since for $\tau \in [t-\plength, t]$, it holds that
  $\PR{L_v(\tau)} > 0.058$, then for a sufficiently large $c$ we have 
  that node $v$ is not competing in round $t$ with high probability.
\end{proof}

Next, we show that if a node hears a beep or produces a beep in a round
when its neighborhood (and its neighbor's neighborhoods) has a
``low'' beep potential, then with constant probability either it joins
the MIS, or one of its neighbors joins the MIS. In the following lemma 
we say a node beeps alone at time $t$, if that node beeped at time $t$ 
and all of its neighbors listened at time $t$.

\begin{lemma}
    Assume that node $v$ beeps or hears a beep in round $t$ and that
    $E_u(t) \le \half$ for every $u \in N(v)\cup\set{v}$.  Then with
    probability at least $\frac{1}{e}$ either $v$ beeps alone, or one
    of its neighbors beeps alone in round $t$.
  \label{lem:goodbeep}
\end{lemma}

\begin{proof}
  For simplicity we rename the set $N(v)\cup \set{v}$ to the set
  $\set{1,\ldots,k}$ where $k=|N(v)|+1$.
  For $i \in [k]$ we consider three events.
  \begin{align*}
    A_i&: \text{ Node $i$ beeps in round $t$.} \\
    B_i&: \text{ Node $i$ beeps alone in round $t$.} \\
    S &: \bigcup_{i \in [k]} B_i
  \end{align*}
  Our aim is to show that the event $S$ happens with constant
  probability. Fix $i \in [k]$, as a first step we show that 
  $\PR{B_i|A_i}$ is constant.

  \begin{align*}
  \PR{B_i|A_i} & =\PR{\overline{\bigcup_{w \in N(i)} A_w}} =
  \PR{\bigcap_{w \in N(i)} \overline{A_w}} \\&= \prod_{w \in N(i)}(1-b_w(t))
  \\
  & \ge \exp\paren{-2\sum_{w \in N(i)} b_w(t)} = e^{-2 E_i(t)}
  \end{align*}

  Moreover, since by assumption $E_i(t)\le \half$, it follows that 
  $\PR{B_i|A_i}\ge \frac{1}{e}$.

  We define the following finite partition of the probability space:
  \begin{align*}
    \xi_1 & = A_1,\\
    \xi_2 & = A_2 \cap \neg{A_1},\\
    \xi_3 & = A_3 \cap \neg{A_2} \cap \neg{A_1},\\
    \ldots &\\
    \xi_k & = A_k \cap \bigcap_{i=1}^{k-1} \neg{A_i}.
  \end{align*}
  Recall that by assumption our probability space is conditioned on the 
  event that ``node $v$ beeps or hears a beep in round $t$'', or in 
  other words $\exists i \in [k]$ such that $A_i$ has occurred.  
  Moreover, observe that $\bigcup_{i=1}^k \xi_i = \bigcup_{i=1}^k A_i$, 
  and thus $\PR{\bigcup_{i=1}^k \xi_i}=1$.

  Since the events $\xi_1,\ldots,\xi_k$ are pairwise disjoint, by the
  law of total probability we have:
  $$\PR{S} = \sum_{i=1}^k \PR{S|\xi_i}\PR{\xi_i}$$
  Finally since $\PR{S|\xi_i}\geq\PR{B_i|\xi_i} \ge \PR{B_i|A_i}
  \ge \frac{1}{e}$ then $\PR{S} \ge \frac{1}{e} \sum_{i=1}^k
  \PR{\xi_i} = \frac{1}{e}$.
\end{proof}

The three previous lemmas give us the ingredients
necessary to prove Lemma~\ref{lem:bterminate} and thus complete the proof of
Theorem~\ref{thm:simpleMIS}:

\begin{proof}[Proof of Lemma~\ref{lem:bterminate}\\]
  We say a node has an event in round $t$, if it beeps or hears a beep
  in round $t$. First we claim that a node has an
  event every $\bigO(\log^2 N)$ rounds.
  Consider a node that does not hear a beep within $\bigO(\log^2 N)$ 
  rounds (if it does hear a beep, the claim clearly holds). Then after 
  $\bigO(\log^2 N)$ it will join the MIS and beep and the claim follows.

  From Lemma~\ref{lem:lowEu} we know that when a node decides to beep, with
  high probability the beep potential of its neighborhood is less than
  $\half$. We can use a union bound to say that when a node hears a
  beep, with high probability the beep was produced by a node with a
  beep potential less than $\half$.
  Therefore, we can apply Lemma~\ref{lem:goodbeep} to say that with constant
  probability every time a node has an event, either the node joins the
  MIS (if it was not in the MIS already) or it becomes covered by an MIS
  node.

  Hence, with high probability after $\bigO(\log n)$ events, a node is
  either part of the MIS or it becomes covered by an MIS node. Since
  there is an event every $\bigO(\log^2 N)$
  rounds, this implies that with high probability a node is either
  inside the MIS or has a neighbor in the MIS after $\bigO(\log^2 N
  \log n)$ rounds.
\end{proof}

This completes the proof for Theorem~\ref{thm:simpleMIS}.


  \section{Wake-on-Beep and Sender-Side\\ Collision Detection} \label{sec:Nmis}

This section considers a different relaxation of the beeping model.  
Specifically, while still allowing the adversary to wake up
nodes arbitrarily, in this and the next section we assume that
sleeping nodes also wake up upon receiving a beep. We call this the 
wake-on-beep
assumption. Moreover, in this section we also assume that when a node beeps, it
receives some feedback from which it can infer if it beeped alone, or
if one of its neighbors beeped concurrently. We call this sender-side 
collision detection. We will show that in the wake-on-beep model
with sender-side collision detection it is possible to locally converge 
to an MIS in $O(\log^2 n)$ time, even if nodes have no knowledge about 
the network topology, including its size.

This algorithm is an improvement of the algorithm presented in
\cite{science11}, which used an upper bound on the size of the
network. In this algorithm nodes go through several iterations in
which they gradually decrease the probability of being selected.
The running time of the algorithm is still $O(\log^2 n)$ as we show
below. Compared to the algorithm in \cite{science11}, in addition to
eliminating the dependence on any topological information, the current
algorithm tolerates adversarial wake ups if we assume
wake-on-beep.

\paragraph{Algorithm.}

\begin{algorithm}[t!]
    \caption{Wake-on-beep and sender-side collision
        detection}
    \label{alg:mis1}
    \begin{algorithmic}[1]
        \State {\bf upon} waking up (by adversary or beep)
        \State \quad {\bf do}
        beep to wake up neighbors
        \State wait for 1 round; $x\leftarrow 0$ \Comment{while neighbors wake
            up}
        \Repeat
        \State $x\leftarrow x+1$ \Comment{$2^x$ is current size estimate}
        \For{$i\in \set{0,\dots,x}$} \Comment{$\log 2^x$ phases}
        \State \hspace{2ex} {\bf ** exchange 1 ** with 3 rounds}
        \State listen for 1 round; $v\leftarrow 0$ \Comment{round 1}
        \State w/probability $1/2^i$, beep and set $v\leftarrow 1$
        \Comment{round 2}
        \State listen for 1 round \Comment{round 3}
        \State {\bf if} received beep in  exchange 1 {\bf
            then} $v \leftarrow 0$
        \State \hspace{2ex} {\bf ** exchange 2 ** with 3 rounds}
        \State listen for 1 round \Comment{round 1}
        \State {\bf if} $v=1$ {\bf then} beep and join MIS
        \Comment{round 2}
        \State listen for 1 round \Comment{round 3}
        \EndFor
        \Until in MIS or received beep in exchange 2
        \State terminate
    \end{algorithmic}
\end{algorithm}

The algorithm proceeds in phases each consisting of $x$ steps where
$x$ is the total number of phases performed so far (the phase
counter). 
Step $i$ of
each phase consists of two exchanges. In the first exchange nodes beep
with probability $p_i$ (the value of $p_i$ is given by the algorithm), and in 
the second exchange a node that
beeped in the first exchange and did not hear a beep from any of its
neighbors, beeps again, signaling its neighbors it has joined the MIS
and they should become inactive and exit the algorithm.

Nodes that are woken up by the adversary propagate a wave of wake-up beeps 
throughout the network. Upon
hearing the first beep, which must be the wake up beep, a node
broadcasts the wake up beep in the next round, and then waits one
round to ensure none of its neighbors are still asleep. This
ensures that all neighbors of a node wake up either in the same round
as that node or one round before or after that node. Due to these
possible differences in wakeup time, we divide each exchange into
$3$ rounds. During the second round of the first exchange each active 
node
beeps with probability $p_i$ (the
value of $p_i$ is given in the algorithm). The second exchange also
takes three rounds. A node that beeps in the
first exchange joins the MIS if none of its neighbors beeped in any of 
the three rounds of the first exchange. Such
a node again  beeps in the second round of the second
exchange signaling its neighbors to terminate the algorithm.
The algorithm is detailed in Algorithm~\ref{alg:mis1}.


\paragraph{Safety.}
While the algorithm in \cite{science11} uses a different set of coin
flip probabilities, it relies on a similar two exchanges structure
to guarantee the safety properties of the algorithm.
In \cite{science11}, each exchange is only one round (since synchronous wakeup 
is assumed). We thus need to show that replacing each one round exchange 
with a three round exchange does not affect the MIS safety properties 
of~\cite{science11}. We start by proving that the termination lemma from 
\cite{science11}, which relies on the fact that all neighbors are using 
the same probability distribution in each exchange, still holds.

\begin{lemma}\label{lem:mis1}
All messages received by node $j$ in the first exchange of
step $i$ were sent by processes using the same probability as $j$
in that step.
\end{lemma}

\begin{proof} Let $k$ be a neighbor of
$j$. If $k$ started in the same round as $j$ (both woke up at the
same round) then they are fully synchronized and we are done. If $k$
started before $j$ then the first message $k$ sent has awakened $j$.
Thus, they are only one round apart in terms of execution. Any
message sent by $k$ in the second round of the first exchange of step
$i$ would be received by $j$ in the first round of that exchange.
Similarly, if $k$ was awakened after $j$ it must have been a $1$
round difference and $j$ would receive $k$'s message of the first
exchange of step $i$ (if $k$ decided to beep) in the third
round of that exchange. Thus, all messages received by $j$ are from
processes that are also in step $i$ and so all processes from
which $j$ receives messages in that exchange are using the same
probability distribution.
\end{proof}

A similar argument would show that all messages received
in the second exchange of step $i$ are from processes that are in
the second exchange of that step. Since our safety 
proof~\cite{science11} only relies
on the coherence of the exchange it still holds for this algorithm.

\paragraph{Termination.}
After establishing the safety guarantees, we next prove that with
high probability all nodes terminate the algorithm in $O(\log^2 n)$
time where $n$ is the number of nodes that participate in the
algorithm. Let $d_v$ be the number of {\em active} neighbors of node
$v$. We start with the following definition of \cite{alon86}. A node
$v$ is \emph{good} if it has at least $d_v/3$ active neighbors
$u$, such that, $d_u \leq d_v$. An edge is \emph{good} if at least one 
of its endpoints is a \emph{good} node. 

\begin{lemma}[Lemma 4.4 from \cite{alon86}]
  \label{lem:goodedges}
  In every graph $G$ at least half of the edges are good. Thus, $ 
  \sum_{v \in good} d_v \geq |E| / 2 $.
\end{lemma}

Note that we need less than $O(\log^2 n)$ steps to reach $x \geq \log
n$, since each phase $x \le \log n$ has less than $\log n$ steps.
When $x \geq \log n$, the first $\log n$ steps in each phase are
using the probabilities: $1,1/2,1/4, ... , 2/n, 1/n$. Below we show
that from round $x = \log n$, we need at most $O(\log n)$ more
phases to guarantee that all processes terminate with high probability.
We say an edge is \emph{deleted} if one of its endpoints joins the MIS.

\begin{lemma}
  In a phase (with more than $\log n$ steps) in expectation a constant 
  fraction of the edges are deleted.
\end{lemma}

\begin{proof}
  Fix a phase $j$, and fix a good node $v$. We claim that the expected 
  number of edges incident to $v$ that are deleted
in phase $j$ is $\Omega(d_v)$. To prove the claim assume that at the 
beginning of
phase $j$,  $2^k \leq d_v \leq 2^{k+1}$ for some $0 < k < \log n$. If
when we reach step $i=k$ in phase $j$ at least $d_v /20$ edges
incident to $v$ were already removed we are done. Otherwise, at
step $i$ there are still at least $d_v/3 - d_v /20 > d_v/4 \geq
2^{k-2}$ neighbors $u$ of $v$ with  $d_u \leq d_v$. Let $A$ be the event that 
node $v$ or a neighbor $u$ with $d_u < d_v$ beeps. Node $v$ and all
its neighbors $u$ are flipping coins with probability $\frac{1}{2^k}
$ at this step and thus the probability of $A$ occurring is:
\[
    \Pr(A) \geq
    1-\left(1-\frac{1}{2^k}\right)^{2^{k-2}} \ge 1-e^{-1/4}.
\]
On the other hand, all such nodes $u$, and $v$, have less than
$2^{k+1}$ neighbors.  Thus, the probability that a node from this
set that beeps does not collide with any other node
is:
\[
\Pr(\text{no collisions}) \geq (1 - \frac{1}{2^k})^{2^{k+1}}
\ge 1/e^4.
\]
Thus, in phase $j$ a node $v$ has probability of at least
$(1-\frac{1}{e^{1/4}})\frac{1}{e^4} \geq \frac{1}{2^8} $ to be
removed. Thus, the probability that $v$ is removed in phase $j$ is 
$\Omega (1)$ and hence the expected number of edges incident with $v$ 
removed during this phase is $\Omega (d_v)$, which completes our claim.

Combining the previous claim with Lemma~\ref{lem:goodedges}, then we can 
use linearity of expectation to show that the expected number of edges 
deleted in each phase is at least $\Omega ( \sum_{v\in good} d_v) = 
\Omega( | E|)$.
\end{proof}

With this lemma in place, we are ready to prove the main theorem of this 
section.

\begin{theorem}
  Using sender-side collision detection and wake-on-beep, Algorithm 2 
  locally converges to an MIS in $O(\log^2 n)$ rounds.
\end{theorem}

\begin{proof}
Note that since the number of edges removed in a phase in a graph
$(V,E)$ is clearly always at most $|E|$, the last lemma implies that
for any given history, with probability at least $\Omega(1)$, the
number of edges removed in a phase is at least a constant fraction
of the number of edges that have not been deleted yet. Therefore
there are two positive constants $p$ and $c$, so that the probability 
that in a phase at least a fraction
$c$ of the number of remaining edges are deleted is at least $p$.
Call a phase successful if at least a fraction $c$ of the remaining
edges are deleted during the phase.

By the above reasoning, the probability of having at least $z$
successful phases among $m$ phases is at least the probability that
a binomial random variable with parameters $m$ and $p$ is at least
$z$. By the standard estimates for binomial distributions, and by
the obvious fact that $O( \log |E|/c)=O(\log n)$, starting from $x =
\log n$ we need an additional $O(\log n)$ phases to finish the
algorithm. Since each of these additional $O(\log n)$ phases
consists of $O(\log n)$ steps, and since as discussed above until $x
= \log n$ we have less than $O(\log^2 n)$ steps, the total running
time of the algorithm is $O(\log^2 n)$.
\end{proof}

%
%
%


\section{Wake-on-Beep Without Sender-Side\\ Collision Detection}
\label{sec:beepmis}

In the previous section we assumed that nodes are endowed with 
sender-side collision detection and can thus tell whether one of their 
neighbors beeped even in rounds in which they beep. In this section we 
remove this assumption and present an algorithm for the wake-on-beep 
model that locally converges to an MIS without using sender-side 
collision detection.

\paragraph{Algorithm.}

To extend Algorithm~\ref{alg:mis1} to a model with no collision detection
we increase the number of exchanges in each step from 2 to $c x$ where 
$c$ is a constant derived below and $x$
is the same as in Algorithm~\ref{alg:mis1} and represents the current 
estimate of the network size. Each series of $c x$ rounds simulates with 
high probability an exchange with sender-side collision detection.
Prior to starting the exchanges in each step each {\em active} process 
flips a coin with the same probability as in Algorithm~\ref{alg:mis1}.  
If the flip outcome is $0$ (tails) the process only listens in the next 
$cx$ exchanges (for a constant $c$ discussed below). If the flip outcome 
is $1$ the process sets $v=1$ and picks each entry in the vector $X$ of 
length $cx$ to be $1$ or $0$ independently and uniformly at random.  
Following this, the process picks one entry in the vector $X$ 
independently and uniformly at random and sets it to $1$ (this is only 
to guarantee that at least one entry in $X$ is equal to one).
In exchange $j$ of every phase, a process beeps if $X(j)=1$ and listens 
if $X(j)=0$.
If at any of the exchanges it listens and hears a beep it sets $v=0$ and 
stops beeping (even in the selected exchanges). If a node hears a beep 
during these exchanges it does not exit the algorithm.  Instead, it 
denotes the fact that one of its neighbors beeped and sets itself to be 
inactive. If it does not hear a beep in any of the exchanges of the 
following phase it becomes active and continues as described above.  
Similarly, a node that beeped and did not hear any beep in a specific 
step (indicating that it can join the MIS) continues to beep 
indefinitely (by selecting half the exchanges in all future steps to 
beep in them).


\algrenewcommand\algorithmicloop{\textbf{repeat forever}}
\begin{algorithm}[t]
    \caption{Wake-on-beep \emph{without} sender-side collision
        detection}
    \label{alg:mis2}
    \begin{algorithmic}[1]
        \State {\bf upon} waking up (by adversary or beep)
        \State \quad{\bf do} beep to wake up neighbors
        \State wait for 1 round;  \Comment{while neighbors wake up}
        \State $x\leftarrow 0$; $v\leftarrow 0$; $z\leftarrow 0$
        \Loop
        \State $x\leftarrow x+1$ 
        \For{$i\in \set{0,\dots,x}$} 
        \State {\bf if} $v=0 \land z=0$ {\bf then} $v\leftarrow 1$ w/probability $1/2^i$
        \State $X \leftarrow$ random $0/1$-vector of length $cx$
        \State $z\leftarrow 0$
        \State \hspace{2ex} {\bf ** \boldmath$cx$ competition
            exchanges **}
        \For{$k\in\set{1,\dots,cx}$}
        \State listen for 1 round
        \State {\bf if} beep received {\bf then} $v\leftarrow0$; $z\leftarrow 
        1$ 
        \If{$v=0\lor X[k]=0$}
        \State listen for 1 round;
        \State {\bf if} beep received {\bf then} $v\leftarrow 0$; $z\leftarrow 1$
        \Else 
        \State beep for 1 round
        \EndIf
        \State listen for 1 round
        \State {\bf if} beep received {\bf then} $v\leftarrow 0$
        \EndFor
        \EndFor
        \EndLoop
    \end{algorithmic}
\end{algorithm}

We say a process $u$ is \emph{in conflict with a neighbor} $v$ if both 
have $v=1$. We say a process $u$ is \emph{in conflict} if it is in 
conflict with respect to any of its neighbors. 

The main difference between this algorithm and Algorithm~\ref{alg:mis1} 
is the addition of a set of competition exchanges at the
end of each coin flip. The number of competition exchanges is proportional
to the current phase counter (which serves as the current estimate
of the network size). Initially the competition rounds are short and
so they would not necessarily remove all conflicts. We require that
nodes that attempt to join continue to participate in all future
competition rounds (when $v=1$). Processes that detect a MIS member
as a neighbor set $z$ to 1 and do not beep until they go
through one complete set of competition exchanges in which they do
not hear any beep. If and when this happens they set $z = 0$ and
become potential MIS candidates again.

While not all conflicts will be resolved at the early phases, when
$x \geq \log n$ each set of competition exchanges is very likely to
remove all conflicts. We prove below that once we arrive at such
$x$ values, all conflicts are resolved with very high probability
such that only one process in a set of conflicting processes
remains with $v=1$ at the end of these competition exchanges. From
there, it takes another $O(\log n)$ phases to select all members of
the MIS as we have shown for Algorithm 1. Since each such phase
takes $O(\log n)$ steps with each step taking $O(\log n)$ rounds
for the competition,  the total running time of the algorithm is
$O(\log^3 n)$.

\begin{lemma}
  Assume process $y$ is in conflict at step $i$ of phase $x \geq \log 
  n$.  The probability that $y$ remains in conflict at the end of the 
  $cx$ competition exchanges for step $i$ is at most 
  $\frac{1}{n^{c/3}}$.
\end{lemma}

\begin{proof} If at any of the
exchanges in this step all neighbors of $y$ have $v=0$ we are done.
Otherwise in each exchange, with probability at least $1/4$, $y$
decided not to beep whereas one of its conflicting neighbors
decided to beep. Thus, the probability that $y$ remains in conflict
in a specific exchange is at most $3/4$. Since
there are $(c \log n)$ exchanges in this step, the probability that
$y$ is in conflict at the end of these exchanges
is at most $(\frac{3}{4})^{c\log n} \leq \frac{1}{n^{c/3}}$.
\end{proof}

Note that if two nodes remain in conflict after an exchange, they 
continue to beep in the following phase.
As we proved in the previous section, if all conflicts are resolved in
the $O(\log n)$ phases that follow the phase $x = \log n$ the
algorithm will result in a MIS set with very high probability. Since
we only need $O( \log^2 n)<n$ steps for this, and we have $n$ nodes,
the probability that there exists a step and a node in phase $x \geq
\log n$ such that a node that conflicted during this step with a
neighbor does not resolve this conflict in that step is smaller
than $\frac{1}{n^{c/3-2}}$. Thus, with probability at least $1 - 
\frac{1}{n^{c/3 -2}}$  all conflicts are resolved and the MIS
safety condition holds.

We note that the fact that the vector $X$ always contains at least one $1$
guarantees that once an MIS is computed it remains stable forever. We also 
remark that in contrast to the algorithm in Section~\ref{sec:Nmis},
it is not possible for the algorithm in this section to terminate safely at any 
point of time (see Section~\ref{sec:termination} for details).
This discussion completes the main proof of our main theorem for this section:

\begin{theorem}
In the wake-on-beep model, Algorithm~\ref{alg:mis2} locally converges to 
an MIS in $O(\log^3 n)$ rounds.
\end{theorem}

  \section{Synchronized Clocks}\label{sec:synch}

For this section the only assumption we make on top of the beeping model is 
that
that nodes have synchronized clocks, that is, know the current round number $t$. 

The idea of the algorithm is to simulate Luby's permutation 
algorithm~\cite{luby86}.  In Luby's permutation algorithm a node picks a 
random $\bigO(\log n)$-size priority which it shares with its neighbors.  
A node then joins the MIS if it has the highest priority among its 
neighbors, and all neighbors of an MIS node become inactive.
Despite the fact that we describe the algorithm for
the message passing model, it is straightforward to adapt the priority 
comparisons
to the beeping model. For this, a node sends its priority bit by
bit, starting with the highest-order bit and using a \beep for a $1$. The
only further modification is that a node stops sending its priority as soon as 
it hears a \beep on a higher order bit during which it
remained silent because it had a zero in the corresponding bit. Using this
simple procedure, a node can easily realize when a neighboring node has
a higher priority.
Furthermore, nodes which do not hear any \beep correspond to the nodes 
which have the highest-priority in its
neighborhood (strictly speaking, this correspondence is not exact, since 
the algorithm described allows even more nodes to join the MIS that one 
step of Luby, but without violating any safety guarantees).

Therefore, as long as nodes have a synchronous start and know $n$ (or an upper 
bound on $n$) it is straightforward to get Luby's permutation algorithm working 
in the beeping model in $\bigO(\log^2 n)$ rounds.

In the rest of this section we show how to remove the need for an upper
bound on $n$ and a synchronous start. We leverage synchronized
clocks to synchronize the exchanges of priorities amongst neighboring nodes.
Our algorithm keeps an estimate $k$ for the required priority-size 
$\bigO(\log n)$. Whenever two nodes tie for the highest priority the 
algorithm concludes that $k$ is not large enough and doubles its 
estimate.
The algorithm uses a Restart-Bit to ensure that nodes locally work with the 
same estimate $k$ and run in a synchronized manner in which priority 
comparisons start at the same time (namely every $t \equiv 0 \pmod k$). It is 
not obvious that either a similar $k$ or a synchronized priority comparison is 
necessary but it turns out that algorithms without them can stall for a long 
time. In the first case this is because nodes with a too small $k$ repeatedly 
enter the MIS state simultaneously, while in the second case many 
asynchronously competing nodes (even with the same, large enough $k$) keep 
eliminating each other without one becoming dominant and transitioning into the 
MIS state.

\paragraph{Algorithm:}
Nodes have three different internal states: \emph{inactive}, \emph{competing}, 
and \emph{MIS}.  Each node has an estimate $k$ on the priority-size that is 
monotone increasing during the execution of the algorithm. Initially all nodes 
are in the inactive state with $k=6$.

Nodes communicate in beep-triplets, and synchronize by starting a
triplet only when $t \equiv 0 \pmod 3$. The first bit of the triplet is
the Restart-Bit. A \beep is sent in the Restart-Bit if and only if $t
\not\equiv 0 \pmod k$, otherwise a node listens in the Restart-Bit. If a node 
hears a \beep in its Restart-Bit it
doubles its estimate for $k$ and it becomes inactive. The second
bit sent in the triplet is the MIS-Bit. A \beep is sent for the MIS-Bit
if and only if a node is in the MIS state. If a node hears a \beep on
the MIS-bit it becomes inactive. The last bit sent in the triplet is the
Competing-Bit. If inactive, a node listens in the Competing-Bit. If a 
node is competing it
sends a \beep with probability 1/2 in the Competing-Bit. If a node is in 
the MIS state and it listened in the previous Competing-Bit then it 
beeps in the current Competing-Bit. On the other hand if node in the MIS 
state beeped in the previous Competing-Bit, then it flips a coin to 
decide weather to beep or listen in the current Competing-Bit. This 
ensures a node in the MIS state beeps every 2 round.
If a node hears a \beep on
the Competing-Bit it becomes inactive, and if the node was in the MIS-state it 
also doubles its estimate for $k$.
Lastly, a node transitions from inactive to competing (or from competing 
to MIS) between any time $t$ and $t+1$ for $t\equiv 0\pmod k$. The 
pseudo code is described in more detail in Algorithm~\ref{alg:synch}.

\begin{algorithm}
  \caption{Synchronous Clocks.}
  \label{alg:synch}
  \begin{algorithmic}[1]
    \State Initially $state \gets inactive$, $next \gets random~0/1$
    \If{ $t \equiv 0 \mod 3$}
      \Comment{Restart-Bit}
      \If{$t \not\equiv 0 \mod k$}
        \beep
      \Else
        ~listen
        \If {heard beep}
          $state \gets inactive$, $k \gets 2 \cdot k$
        \Else
          ~advance $state$ \State \quad ($inactive\to competing$ or 
          $competing \to MIS$)
        \EndIf
      \EndIf
    \EndIf
    \If{$t \equiv 1 \mod 3$}
      \Comment{MIS-Bit}
      \If{$state = MIS$}
        \beep
      \Else
        ~listen
        \State {\bf if} heard beep {\bf then} $state \gets inactive$
      \EndIf
    \EndIf
    \If{$t \equiv 2 \mod 3$}
      \Comment{Competing-Bit}
      \If{$state = inactve$}
        \State listen
      \ElsIf{$state = competing$}
        \State with probability 1/2 beep, otherwise listen
        \State {\bf if} heard beep {\bf then} $state \gets inactive$
      \ElsIf{$state = MIS$}
        \If{$v = 1$}
          beep, $next \gets random~0/1$
        \Else
          ~listen, $next \gets 1$
          \State {\bf if} heard beep {\bf then} $state \gets inactive$, 
          $k \gets 2\cdot k$
        \EndIf
      \EndIf
    \EndIf
  \end{algorithmic}
\end{algorithm}

\paragraph{Analysis:}

The main result of this section is the following theorem.

\begin{theorem}
  If nodes have synchronous clocks then Algorithm 4 solves the MIS 
  problem in $\bigO(\log^2 n)$ rounds.
\end{theorem}

First, we show that with high probability $k$ cannot become 
super-logarithmic.

\begin{lemma}\label{lem:smallk}
With high probability $k\in \bigO(\log n)$ for all nodes during the 
execution of the algorithm.
\end{lemma}
\begin{proof}
We start by showing that two neighboring nodes in the MIS state must have the 
same estimate $k$ \emph{and} must have transitioned to the MIS state at the 
same time. We prove both parts of this statement by contradiction.

First, suppose by contradiction that two neighboring nodes $u$ and $v$ 
are in the MIS state but $u$ transitioned to this state (the last time) 
before $v$.  In this case $v$ would have received the MIS-bit from $u$ 
and become inactive instead of joining the MIS --  a contradiction.

Similarly, for sake of contradiction, now assume that the neighboring nodes $u$ 
and $v$ are in the MIS state and $k_u < k_v$. In
this case, during the active phase of $u$ before it transitioned to the
MIS at time $t$ it would have hear a beep in its Restart-Bit (produced 
by $v$) and would have switched to the inactive state, which contradicts 
that $u$ is in the MIS state.

We now use this to show that for a specific node $u$ it is unlikely to
become the first node with a too large $k$. For this we note that
$k_u$ is doubled because of a Restart-Bit only if a \beep from a
node with a larger $k$ is received. This node can therefore not be
responsible for $u$ becoming the first node getting a too large
$k$. The second way $k$ can increase is if a node transitions out of
the MIS state because it receives a Competing-Bit from a neighbor
$v$. In this case, we know that $u$ competed against at least one such
neighbor for $k/6$ phases without loosing in any of these phases. The 
probability that
this happens is $2^{-k/6}$. Hence, if $k \in \Theta(\log n)$, then with 
high probability it does not happen.  A union bound over all nodes and 
the polynomial number of rounds in which nodes are not yet stable 
finishes the proof.
\end{proof}

\begin{lemma}\label{thm:synchalg}
    If during an execution the $\bigO(\log n)$ neighborhood of node $u$
    has not changed for $\Omega(\log^2 n)$ rounds then node $u$ is 
    stable with high probability,
    i.e., $u$ is either in the MIS state with all its neighbors being
    inactive or it has at least one neighbor in the MIS state whose
    neighbors are all inactive.
\end{lemma}

\begin{proof}
First observe that if the whole graph has the same value of $k$ and no
two neighboring nodes transition to the MIS state at the same time, then
our algorithm behaves exactly as Luby's original permutation algorithm,
and therefore terminates after $\bigO(k \log n)$ rounds with high
probability.
From a standard locality argument, it follows that a node $u$ also
becomes stable if the above assumptions only hold for a $\bigO(k \log 
n)$
neighborhood around $u$. Moreover, since Luby's algorithm performs
only $\bigO(\log n)$ rounds in the message passing model, we can improve 
our
locality argument to show that in if a $\bigO(\log n)$ neighborhood 
around
$u$ is well-behaved, then $u$ behaves as in Luby's algorithm.

Since the values for $k$ are monotone increasing and propagate between
two neighboring nodes $u$ and $v$ with different $k$ (i.e., $k_u > k_v$) 
in at
most $2k_u$ steps, it follows that for a node $u$
it takes at most $\bigO(k_u \log n)$ rounds until either $k_u$
increases or all nodes $v$ in the $\bigO(\log n)$ neighborhood of $u$ 
have
$k_v = k_u=k$ for at least $\bigO(k \log n)$ rounds. We can furthermore 
assume
that these $\bigO(k \log n)$ rounds are collision free (i.e, no two
neighboring nodes go into the MIS), since any collision leads with high
probability within $\bigO(\log n)$ rounds to an increased $k$ value for 
one
of the nodes. 

For any value of $k$, within $\bigO(k \log n)$ rounds a node thus either
performs Luby's algorithm for $\bigO(\log n)$ priority exchanges, or it
increases its $k$. Since $k$ increases in powers of two and, according
to Lemma~\ref{lem:smallk}, with high probability it does not
exceed $\bigO(\log n)$, after at most $\sum_i^{\bigO(\log \log n)} 
2^i\cdot
3\cdot \bigO(\log n) \in \bigO(\log^2 n)$ rounds the status labeling 
around a
$\bigO(\log n)$ neighborhood of $u$ is a proper MIS. This means that $u$ 
is
stable at some point, and the MIS-bit guarantees that no competing 
neighbor of $u$ will join the MIS and therefore stability is preserved 
for the rest of the execution.
\end{proof}

We remark that as the algorithm of Section~\ref{sec:alg1}, this algorithm is also robust enough to work as-is with an adversary capable of crashing nodes (with the same caveats on the guarantees mentioned in Section~\ref{sec:alg1}).


  \begin{acknowledgements}
We thank the anonymous reviewers for their feedback to improve
the quality of this paper.
Research supported in part by AFOSR Award FA9550-08-1-0159,
NSF Award CNS-1035199,
NSF Award CCF-0937274,
NSF Award CCF-0726514,
ERC advanced grant, USA-Israeli BSF grant, and the Israeli I-Core 
program.
  \end{acknowledgements}

\end{document}